\documentclass[journal]{IEEEtran}

\ifCLASSINFOpdf
\else
   \usepackage[dvips]{graphicx}
\fi
\usepackage{url}

\usepackage{cite}

\usepackage{nicefrac}
\usepackage{amsmath}
\usepackage{amssymb}
\usepackage{graphicx}
\usepackage{booktabs}
\usepackage{hyperref}
\usepackage{soul}
\usepackage{color}

\usepackage{algorithm}
\usepackage{algpseudocode}

\usepackage{booktabs}
\usepackage{makecell}
\usepackage{siunitx}

\newcommand{\pmstd}[1]{{\scriptscriptstyle \pm\, #1}}

\newtheorem{theorem}{Theorem}
\newtheorem{proof}{Proof}

\newcommand{\diff}{\mathrm{d}}
\newcommand{\R}{\mathbb{R}}

\begin{document}

\title{An integration-free approach for particle flow filtering}

\author{Domonkos Csuzdi, Tamás Bécsi, and Olivér Törő
\thanks{The project supported by the Doctoral Excellence Fellowship Programme (DCEP) is funded by the National
Research Development and Innovation Fund of the Ministry of Culture and Innovation and the Budapest
University of Technology and Economics.\\This work was funded by the National Research, Development and Innovation Office (NKFIH) under NKKP Grant Agreement No. ADVANCED 152880}
\thanks{Domonkos Csuzdi, Tamás Bécsi, and Olivér Törő are with the Department of Control for Transportation and Vehicle Systems, Faculty of Transportation Engineering and Vehicle Engineering, Budapest University of Technology and Economics, Műegyetem rkp.~3, Budapest, H-1111, Hungary. (e-mail: domonkos.csuzdi@edu.bme.hu, becsi.tamas@kjk.bme.hu, toro.oliver@kjk.bme.hu).}}

\markboth{Journal of \LaTeX\ Class Files, Vol. 14, No. 8, August 2015}
{Shell \MakeLowercase{\textit{et al.}}: Bare Demo of IEEEtran.cls for IEEE Journals}
\maketitle

\begin{abstract}
Log-homotopy particle flow filters realize nonlinear Bayesian estimation by continuously migrating samples from the prior to the posterior distribution. This transport is governed by a pseudo-time ordinary differential equation (ODE). A major practical challenge of these filters is the need for numerical integration, which suffers from high computational cost and susceptibility to stiffness. This paper develops an exact, integration-free closed-form solution for the exact Daum--Huang deterministic particle flow under vector linear Gaussian measurements. By transforming the ODE into a specific eigenspace, we derive closed-form algebraic expressions for both the homogeneous state transition matrix and the inhomogeneous forcing term. We prove that this analytic solution is equivalent to the exact Kalman measurement update.
We embed this closed-form evaluation within an $N$-step piecewise method for nonlinear measurement models. We further propose a constant contraction rate substep schedule that equalizes the per-step contraction along the eigendirection of $D$ associated with the largest eigenvalue $\alpha_{\max}$. The result is a stiffness-mitigating, integration-free particle update for highly nonlinear measurement models. On a bearings-only tracking benchmark, it achieves the lowest error among the compared filters, at a per-update cost comparable to deterministic particle flow baselines and substantially lower than stochastic flows.
\end{abstract}


\IEEEpeerreviewmaketitle

\section{Introduction}

\IEEEPARstart{P}{article} filters are widely used for nonlinear Bayesian state estimation~\cite{chen2003bayesian,Wang2026, PATWARDHAN2012933, Gordon1993Novel}, but they inherently suffer from weight degeneracy in high-dimensional or sharp likelihood scenarios~\cite{bickel2008sharp}. To address this, particle flow filters were introduced to replace traditional importance resampling with a continuous transport of probability mass~\cite{Daum2007nonlinear,daum2015proof,dai2021stability}. 

Log-homotopy particle flow filters realize the Bayesian update by moving samples from the prior density $g(x)$ to the posterior $p(x)$ continuously over a pseudo-time interval $\lambda \in [0, 1]$. A path of intermediate densities $p(x, \lambda)$ is constructed using the linear log-homotopy
\begin{equation} \label{eq:homotopy}
    \log p(x, \lambda) = \log g(x) + \lambda \log \ell(x) - \log \kappa(\lambda) \, ,
\end{equation}
where $\ell(x)$ is the likelihood and $\kappa(\lambda)$ is the instantaneous normalization factor. This continuous deformation induces a differential equation that governs the flow of the particles. 

Various flow formulations have been derived so far, both deterministic and stochastic variants~\cite{dai2021new,crouse2020consideration,mori2022note}. Among these, the exact Daum--Huang (EDH) filter~\cite{daum2010exact, ding2012implementation, choi2011discussion} is highly attractive due to its deterministic nature. However, a major practical limitation of particle flow filters is that the resulting ordinary differential equation (ODE) is typically solved via numerical integration. This introduces discretization errors, high computational cost, and susceptibility to stiffness~\cite{dai2022stiffness}.

Stiffness in the particle flow ODEs poses a long-recognized challenge for accurate numerical evaluation~\cite{daum2016plethora}.
Because there is no universally accepted mathematical definition of stiffness, in practice it is typically characterized heuristically through the ill-conditioning of the flow Jacobian~\cite{daum2014seven}.
Several strategies exist to mitigate stiffness. One direction adapts the discretization scheme to the flow. Mori et al.~\cite{mori2016adaptive} proposed an adaptive Euler step that limits per-particle displacement, later generalized by Dai and Daum~\cite{dai2022stiffness}, who replaced the straight-line log-homotopy with an optimized path derived from a boundary-value problem that minimizes the Hessian condition number.
A second direction keeps an equidistant integration grid but replaces explicit integrators with implicit or stiff-aware schemes~\cite{crouse2020consideration, khan2015improvements}.
A third direction treats the diffusion matrix $Q$ as a design parameter. Initial proposals included minimum stiffness as one of the optimization criteria for $Q$~\cite{daum2016seven}. Subsequent studies derived optimal choices for $Q$ to reduce the Jacobian condition number while satisfying variance and stability constraints~\cite{dai2022role, dai2022design,zhang2025importance, daum2024bullet}.
The present work belongs to a fourth, conceptually distinct category. It sidesteps the stiffness problem entirely by eliminating numerical integration.
In prior works, Crouse~\cite{crouse2021particle} obtained an analytic solution for the geodesic flow and reduced the Gromov flow to non-stiff stochastic integrals, while Törő and Bécsi~\cite{toro2023analytic} proved that the EDH flow is driven by a commutative system and proposed an analytic solution for the scalar measurement case. 

The contributions of this paper are threefold. First, we derive an exact, integration-free evaluation of the EDH particle flow update for the general vector measurement case, transforming the ODE solution into algebraic operations involving an eigendecomposition. Second, we embed this closed-form evaluation within an $N$-step piecewise method for nonlinear filtering, and propose a constant contraction rate (ccr) substep schedule that equalizes the per-step contraction along the eigendirection of $D$ associated with the largest eigenvalue $\alpha_{\max}$, eliminating the stiffness bottleneck of a uniform grid. Third, we prove that the closed-form update is equivalent to the exact Kalman measurement update.
Our implementation and code for reproducing the results are available at \url{https://github.com/DomonkosCs/NAEDH}.

\section{Problem formulation: the EDH flow}

For linear Gaussian measurement models, the exact deterministic flow of a particle $x$ in the EDH framework is governed by a first-order, non-autonomous linear ODE
\begin{equation} \label{eq:flow_ode}
    \frac{d x(\lambda)}{d\lambda} = A(\lambda)x(\lambda) + b(\lambda) \, , \qquad \lambda \in [0,1] \, .
\end{equation}
Here, the homogeneous flow matrix $A(\lambda) \in \mathbb{R}^{n_x \times n_x}$ and the inhomogeneous forcing vector $b(\lambda) \in \mathbb{R}^{n_x}$ are given by~\cite{daum2010exact}:
\begin{align}
    A(\lambda) &= -\tfrac{1}{2} P H^\top \bigl(\lambda H P H^\top + R\bigr)^{-1} H, \label{eq:A_matrix} \\
    b(\lambda) &= \bigl(I + 2\lambda A(\lambda)\bigr) \Bigl[ \bigl(I + \lambda A(\lambda)\bigr) P H^\top R^{-1} z + A(\lambda)\bar{x} \Bigr], \label{eq:b_vector} \nonumber
\end{align}
where $P \in \mathbb{R}^{n_x \times n_x}$ is the prior covariance, $H \in \mathbb{R}^{n_z \times n_x}$ is the measurement matrix, $R \in \mathbb{R}^{n_z \times n_z}$ is the measurement noise covariance, $z \in \mathbb{R}^{n_z}$ is the observation, and $\bar{x} \in \mathbb{R}^{n_x}$ is the prior mean. We assume throughout that $P$ and $R$ are symmetric positive definite and that $H$ has full row rank. The dimension of the identity matrix $I$ follows from context.

Because $A(\lambda)$ commutes with its own integral, the formal exact solution to the ODE~\eqref{eq:flow_ode} takes the form
\begin{equation} \label{eq:formal_solution}
    x(\lambda) = \Phi(\lambda, 0)x(0) + \int_0^\lambda \Phi(\lambda, \mu)b(\mu) \, d\mu \, ,
\end{equation}
where $\Phi(\lambda, \mu) = \exp \bigl( \int_\mu^\lambda A(\tau) \, d\tau \bigr)$ is the state transition matrix.

Törő and Bécsi~\cite{toro2023analytic} previously derived the solution for the scalar measurement case. Here, we derive the solution for the general vector measurement case.

\section{Integration-free exact flow}

\subsection{The state transition matrix}

To decouple the pseudo-time parameter $\lambda$ from the matrix inverse in $A(\lambda)$, we first apply a whitening transformation to the measurement space using $R^{\nicefrac{-1}{2}}$, the inverse of the unique symmetric positive definite square root of the measurement noise covariance $R$.
By factoring $R^{\nicefrac{1}{2}}$ out of both sides of the inner inverse term, we can rewrite $(\lambda H P H^\top \!\!+\! R )^{-1}$ as
\begin{align} \label{eq:inverse_factorization}
    R^{\nicefrac{-1}{2}} ( \lambda R^{\nicefrac{-1}{2}} H P H^\top \!R^{\nicefrac{-1}{2}} + I )^{-1} \!R^{\nicefrac{-1}{2}} \, .
\end{align}

By defining the symmetric positive definite measurement-space matrix $D=R^{\nicefrac{-1}{2}} H P H^\top R^{\nicefrac{-1}{2}} \in \mathbb{R}^{n_z \times n_z}$
we can write
\begin{equation} \label{eq:A_matrix_D}
    A(\lambda) = -\tfrac{1}{2} \bigl(P H^\top R^{\nicefrac{-1}{2}}\bigr) \bigl(\lambda D + I\bigr)^{-1} \bigl(R^{\nicefrac{-1}{2}} H\bigr) \, .
\end{equation}

To evaluate the state transition matrix $\Phi(\lambda, 0) = \exp( \int_0^\lambda A(\mu) d\mu )$, we integrate $A(\mu)$.
Because $(I+\lambda D)^{-1}$ exists for all $\lambda \in [0,1]$, we can recognize the integral representation of the logarithm of a matrix~\cite{wouk1965integral}, that is
\begin{equation}
    \log (\lambda D + I) = \int_0^{\lambda} D (\mu D + I)^{-1} \diff \mu \, ,
\end{equation}
which yields
\begin{equation} \label{eq:matrix_log_integral}
    \int_0^\lambda \!A(\mu) \, d\mu = -\dfrac{1}{2} P H^\top R^{\nicefrac{-1}{2}} \bigl[ D^{-1} \log(I + \lambda D) \bigr] R^{\nicefrac{-1}{2}} H.
\end{equation}
With the introduction of $K = -\nicefrac{1}{2} P H^\top R^{\nicefrac{-1}{2}}$ and $L(\lambda) = D^{-1} \log(I + \lambda D) R^{\nicefrac{-1}{2}} H$, we can factor \eqref{eq:matrix_log_integral} as $KL$. The state transition matrix now reads as $\Phi(\lambda,0) = \exp (KL)$.
To evaluate $\Phi(\lambda,0)$ we will use the matrix exponential identity
\begin{equation}
    \exp(KL) = I + K \left[ (\exp(LK) - I)(LK)^{-1} \right] L \, ,
    \label{eq:low_rank_exp_identity}
\end{equation}
since $K \in \R^{n_x\times n_z}$ and $L \in \R^{n_z\times n_x}$.
The term $LK$ simplifies:
\begin{align}
    L(\lambda)K
      &= D^{-1} \log(I + \lambda D)\, R^{\nicefrac{-1}{2}} H
         \left( -\tfrac{1}{2} P H^\top R^{\nicefrac{-1}{2}} \right) \nonumber \\ 
      &= -\tfrac{1}{2}\, D^{-1} \!\log(I \!+\! \lambda D)\!\,
          \underbrace{R^{\nicefrac{-1}{2}} H P H^\top R^{\nicefrac{-1}{2}}}_{D} \, .
\end{align}
Since $D$ commutes with any function of itself, in particular with $\log(I + \lambda D)$, we can regroup the factors:
\begin{equation} \label{eq:easyLK}
    L(\lambda)K = -\tfrac{1}{2}\, D^{-1} D\, \log(I + \lambda D)
                 = -\tfrac{1}{2}\, \log(I + \lambda D) \, .
\end{equation}
Exponentiating the product \eqref{eq:easyLK} cancels the logarithm:
\begin{equation}
    \exp\bigl(L(\lambda)K\bigr)
      = (I + \lambda D)^{\nicefrac{-1}{2}} \, .
\end{equation}

Substituting these results into the low-rank formula \eqref{eq:low_rank_exp_identity}, for $\Phi(\lambda,0)$ we obtain
\begin{multline}
I + K\! \left[-2 \bigl((I + \lambda D)^{\nicefrac{-1}{2}} \!\!- I\bigr)
               \bigl(\log(I + \lambda D)\bigr)^{-1}
         \right] \!L(\lambda) \, .
\end{multline}
Since $L$ contains $\log(I + \lambda D)$, the state transition matrix takes the algebraic form
\begin{align} \label{eq:Phi_cancellation}
    \Phi(\lambda, 0) &= I \!+\! P H^\top \!R^{\nicefrac{-1}{2}} \left[ (I \!+\! \lambda D)^{\nicefrac{-1}{2}} \!- I \right] D^{-1} R^{\nicefrac{-1}{2}} H \, .
\end{align}

Because $D$ is symmetric positive definite, it admits an eigendecomposition $D = V \Lambda V^\top$ where $V \in \mathbb{R}^{n_z \times n_z}$ contains the orthogonal eigenvectors and
$\Lambda = \mathrm{diag}(\alpha_1,\dots,\alpha_{n_z})$ collects the positive eigenvalues of $D$.

By writing $(I + \lambda D)^{\nicefrac{-1}{2}} = V (I + \lambda \Lambda)^{\nicefrac{-1}{2}} V^\top$ and defining the mappings $E = P H^\top R^{\nicefrac{-1}{2}} V \in \mathbb{R}^{n_x \times n_z}$ and $F^\top = V^\top R^{\nicefrac{-1}{2}} H \in \mathbb{R}^{n_z \times n_x}$ we find
\begin{multline} 
    \Phi(\lambda,0) = I + E \big[ (I + \lambda \Lambda)^{\nicefrac{-1}{2}} - I \big] \Lambda^{-1} F^\top \, .
\end{multline}
The product of the mappings reconstructs the diagonal eigenvalue matrix: $F^\top E = V^\top D V = \Lambda$. 

The middle factor can be written as a diagonal matrix $\Omega(\lambda)$ whose entries are simple scalar functions of the eigenvalues:
\begin{equation}
    \omega_{i}(\lambda)
        = \frac{(1 + \lambda \alpha_i)^{\nicefrac{-1}{2}} - 1}{\alpha_i} \, ,
    \qquad i = 1,\dots,n_z \, .
\end{equation}
With this notation we obtain the compact homogeneous state transition matrix
\begin{equation} \label{eq:Phi_final_compact}
    \Phi(\lambda,0)
      = I + E\, \Omega(\lambda)\, F^\top \, .
\end{equation}

\subsection{The inhomogeneous integral}

Using the mapping matrices $E$ and $F^\top$, the homogeneous flow matrix factors as $A(\mu) =-\tfrac{1}{2} E \Gamma(\mu) F^\top$, where $\Gamma(\mu)=(I+\mu \Lambda)^{-1}$.
We project the measurement $z \in \mathbb{R}^{n_z}$ into the $n_z$-dimensional eigenspace of $D$ as $\tilde{z} = V^\top R^{\nicefrac{-1}{2}} z \in \mathbb{R}^{n_z}$ and similarly the prior mean $\bar{x}$ as $\tilde{x} = V^\top R^{\nicefrac{-1}{2}} H \bar{x} \in \mathbb{R}^{n_z}$.

After we substitute these eigenspace projections into the forcing term $b(\mu)$, the matrix algebra simplifies. Because $A(\mu)E = -\tfrac{1}{2}E\Gamma(\mu)\Lambda$, the projection matrix $E$ factors out from the left side of the entire expression, reducing the vector equation to $n_z$ independent scalar equations: $b(\mu) = E \beta(\mu)$. The $i$-th component of $\beta(\mu)$ is a function of the eigenvalue $\alpha_i$:
\begin{equation}
    \beta_i(\mu) = \frac{ \tilde{z}_i(1 + \tfrac{1}{2} \mu \alpha_i) - \tfrac{1}{2} \tilde{x}_i }{ (1 + \mu \alpha_i)^2 } \, .
\end{equation}

To integrate this forcing term, we apply the mapped state transition matrix. Expanding \eqref{eq:Phi_final_compact} relative to pseudo-time $\mu$, the transition matrix satisfies $\Phi(\lambda, \mu) E = (I + E \Omega(\lambda, \mu) F^\top) E = E(I + \Omega(\lambda, \mu) \Lambda) \equiv E \Psi(\lambda, \mu)$, where $\Psi(\lambda, \mu)$ is a diagonal matrix with entries
\begin{equation}
    \Psi_{ii}(\lambda, \mu) = \left( \frac{1 + \mu \alpha_i}{1 + \lambda \alpha_i} \right)^{\!\nicefrac{1}{2}} .
\end{equation}
The integral $\int_0^\lambda \Phi(\lambda, \mu) b(\mu) d\mu$ thus decouples into independent scalar integrals for each eigendimension $i$:
\begin{equation} \label{eq:scalar_integral}
    c_i(\lambda) = \int_0^\lambda \Psi_{ii}(\lambda, \mu) \beta_i(\mu) \, d\mu \, .
\end{equation}
The integrand in \eqref{eq:scalar_integral} is an algebraic function that admits an exact, closed-form antiderivative. Evaluating this definite integral at the terminal pseudo-time $\lambda = 1$ bypasses all numerical quadrature:
\begin{equation} \label{eq:ci_final}
    c_i(1) = \frac{\alpha_i \tilde{z}_i + \tilde{x}_i \bigl(1 - \sqrt{1 + \alpha_i}\bigr) }{ \alpha_i (1 + \alpha_i) } \, .
\end{equation}

By defining the vector $c = [c_1(1), \dots, c_{n_z}(1)]^\top$, the final Bayesian update of a particle from the prior to the posterior is computed via a single evaluation:
\begin{equation}
    x(1) = \Phi(1, 0)x(0) + E c(1) \, .
\end{equation}
This completes the integration-free resolution of the exact deterministic flow.

\subsection{Equivalence to the Kalman posterior}

\begin{theorem}
The closed-form EDH flow solution at $\lambda=1$ reproduces the
Kalman filter mean and covariance for linear Gaussian measurement models.
\end{theorem}

\begin{proof}
Let the analytic update be written as
\begin{equation}
x_1 = \Phi x_0 + E c \, ,
\qquad
\Phi = I + E \Omega F^\top .
\end{equation}
Taking expectations and using $\mathbb{E}[x_0] = \bar{x}$ yields
\begin{equation}
\bar{x}_1 = \bar{x} + E\bigl(\Omega \tilde{x} + c\bigr) \, ,
\qquad
\tilde{x} = F^\top \bar{x} \, .
\end{equation}

From the closed-form expressions of $\Omega$ and $c$ in the
eigenbasis of $D$, per eigendimension $i$, $\omega_i \tilde{x}_i + c_i = (1+\alpha_i)^{-1}(\tilde{z}_i - \tilde{x}_i)$ after substituting the closed forms, yielding the matrix identity
\begin{equation}
\Omega \tilde{x} + c
=
(I + \Lambda)^{-1}(\tilde{z} - \tilde{x}) \, .
\end{equation}

Substituting $E = P H^\top R^{\nicefrac{-1}{2}} V$ and using
$V(I+\Lambda)^{-1}V^\top = (I+D)^{-1}$,
the mean becomes
\begin{align}
\bar{x}_1
&=
\bar{x}
+
P H^\top R^{\nicefrac{-1}{2}}
V (I+\Lambda)^{-1} V^\top
R^{\nicefrac{-1}{2}}(z - H\bar{x})
\nonumber\\
&=
\bar{x}
+
P H^\top (H P H^\top + R)^{-1}
(z - H\bar{x}) \, ,
\end{align}
which is exactly the Kalman posterior mean.

For the covariance, using $\Phi = I + E \Omega F^\top$, and the identities $F^\top P F = \Lambda$ and
$V(I+\Lambda)^{-1}V^\top = (I+D)^{-1}$, algebra in the eigenbasis yields
\begin{equation}
\Phi P \Phi^\top
=
P - E (I+\Lambda)^{-1} E^\top.
\end{equation}
Substituting back the definitions of $E$ and $D$ gives
\begin{align}
\Phi P \Phi^\top
&=
P
-
P H^\top R^{\nicefrac{-1}{2}}
(I+D)^{-1}
R^{\nicefrac{-1}{2}} H P
\nonumber\\
&=
P
-
P H^\top
(H P H^\top + R)^{-1}
H P \, ,
\end{align}
which matches the Kalman posterior covariance.
\end{proof}

\section{$N$-step piecewise method for nonlinear systems}

While the exact analytic derivations assume a linear measurement matrix $H$ and constant noise covariance $R$, practical tracking problems often involve highly nonlinear observation models $h(x)$. 

To apply the analytic flow to nonlinear models, we partition the pseudo-time interval into $N$ discrete subintervals $0 = \lambda_0 < \lambda_1 < \dots < \lambda_N = 1$~\cite{toro2023analytic}, yielding the $N$-step analytic EDH (NAEDH) algorithm. At the start of each interval $[\lambda_{k-1}, \lambda_k]$, we linearize the observation model around the current particle-cloud mean to obtain a local Jacobian $H_k$.

Treating $H_k$ as constant over the subinterval, we compute the local eigendecomposition $D_k = R^{\nicefrac{-1}{2}} H_k P H_k^\top R^{\nicefrac{-1}{2}} = V_k \Lambda_k V_k^\top$, yielding local projection matrices $E_k$ and $F_k^\top$. Because state transition matrices compose multiplicatively ($\Phi(\lambda_k, 0) = \Phi(\lambda_k, \lambda_{k-1})\Phi(\lambda_{k-1}, 0)$) across intervals, the piecewise homogeneous state transition matrix must be evaluated by exponentiating the integral of the exact local flow matrix directly over $[\lambda_{k-1}, \lambda_k]$. This yields
\begin{equation} \label{eq:n_step_transition}
    \Phi_k(\lambda_k, \lambda_{k-1}) = I + E_k \Omega_k(\lambda_k, \lambda_{k-1}) F_k^\top.
\end{equation}

To simplify the notation of the resulting closed-form evaluations, let us define the eigenvalue-scaling factors at the boundaries of the $k$-th interval as $s_{i,k} = \sqrt{1 + \lambda_k \alpha_{i,k}}$ and $s_{i,k-1} = \sqrt{1 + \lambda_{k-1} \alpha_{i,k}}$. The relative scaling for each eigenvalue $\alpha_{i,k}$ evaluates to the ratio
\begin{equation}\label{eq:omega_step}
    \Omega_{ii,k}(\lambda_k, \lambda_{k-1}) = \frac{ s_{i,k-1} / s_{i,k} - 1 }{\alpha_{i,k}} \, .
\end{equation}

Similarly, to find the local incremental forcing term, we evaluate the definite integral of the mapped inhomogeneous vector over the specific interval bounds $[\lambda_{k-1}, \lambda_k]$. The exact closed-form antiderivative for the $i$-th eigendimension evaluated over the $k$-th subinterval is
\begin{equation} \label{eq:ci_step}
    c_{i,k} = \frac{ \alpha_{i,k} \tilde{z}_{i,k} ( \lambda_k s_{i,k-1} - \lambda_{k-1} s_{i,k} ) + \tilde{x}_{i,k} ( s_{i,k-1} - s_{i,k} ) }{\alpha_{i,k} \, s_{i,k}^2 \, s_{i,k-1}}.
\end{equation}

By defining the local forcing vector $c_k = [c_{1,k}, \dots, c_{n_z,k}]^\top$, the exact particle update across the $k$-th subinterval is computed as
\begin{equation} \label{eq:n_step_update}
    x(\lambda_k) = \Phi_k(\lambda_k, \lambda_{k-1}) x(\lambda_{k-1}) + E_k c_k \, .
\end{equation}
A uniform partition $\lambda_k = k/N$ is wasteful when $D$ is ill-conditioned: along its dominant eigenvalue $\alpha_{\max}$, the bulk of the per-substep contraction $s_{i,k-1}/s_{i,k}$ is absorbed in the first interval, and the remaining substeps refine an already converged state. Requiring instead that this contraction be constant across $k$ along $\alpha_{\max}$ yields the geometric grid
\begin{equation} \label{eq:ccr_schedule}
    \lambda_k = \frac{(1+\alpha_{\max})^{k/N} - 1}{\alpha_{\max}} \, , \qquad k = 0, \dots, N \, .
\end{equation}
We refer to \eqref{eq:ccr_schedule} as the constant contraction rate (ccr) schedule. Although $D_k$ changes across substeps as $H_k$ is re-linearized, we precompute the schedule once from the initial $D$ at the prior mean, on the working assumption that this carries the dominant stiffness.
Our proposed NAEDH algorithm with the ccr schedule is summarized in Alg.~\ref{alg:naedh}. 

Concentrating substeps in the resulting high-contraction region lets NAEDH with ccr outperform adaptive Euler at the same and even substantially larger substep budgets (cf.~Fig.~\ref{fig:three-panel}b).

\begin{algorithm}[t]
\footnotesize
\caption{NAEDH-ccr update.}
\label{alg:naedh}
\begin{algorithmic}[1]
\State \textbf{Input:} $\bar{x}, P, \{x_0^{(j)}\}_{j=1}^{N_p}, z, h(\cdot), R, N$
\State $H\!\gets\!\nabla h(\bar{x})$;\; $D \gets R^{\nicefrac{-1}{2}} H P H^\top R^{\nicefrac{-1}{2}}$
\State Eigendecompose $D = V \Lambda V^\top$, $\Lambda = \mathrm{diag}(\alpha_1, \dots, \alpha_{n_z})$
\State ccr schedule $\{\lambda_k\}_{k=0}^N$ via \eqref{eq:ccr_schedule}
\For{$k = 1$ to $N$}
    \State \textbf{if} $k>1$: re-linearize $H, D, V, \Lambda$ at mean of $\{x_{k-1}^{(j)}\}_{j=1}^{N_p}$
    \State $E \gets P H^\top R^{\nicefrac{-1}{2}} V$;\; $F^\top \gets V^\top R^{\nicefrac{-1}{2}} H$
    \State $\tilde{z} \gets V^\top R^{\nicefrac{-1}{2}} z$;\; $\tilde{x} \gets F^\top \bar{x}$
    \State Form diagonal $\Omega_k$ via \eqref{eq:omega_step} and vector $c_k$ via \eqref{eq:ci_step} over $[\lambda_{k-1}, \lambda_k]$
    \State $x_k^{(j)} \gets (I + E\Omega_k F^\top)\, x_{k-1}^{(j)} + E c_k$ for $j = 1, \dots, N_p$
\EndFor
\State \textbf{Return:} $\{x_N^{(j)}\}_{j=1}^{N_p}$
\end{algorithmic}
\end{algorithm}
\section{Experiments}

Our experimental setup is adapted from~\cite{dai2022design}, which provides the full description of the planar bearings-only tracking scenario and the 6-D constant-acceleration target dynamics (Fig.~\ref{fig:three-panel}a). We modify three aspects to stress-test the filter. First, we superimpose a sinusoidal $y$-perturbation (amplitude 2\,m, period 6\,s) on the otherwise constant-acceleration ground truth, introducing a deliberate mismatch with the filter's dynamic model. Second, the prior mean is intentionally biased by $(-10, 15)$\,m from truth with a loose $P_0 = \mathrm{diag}(200\, I_2, 10\, I_2, I_2)$; combined with the strongly nonlinear bearing model, this deliberately induces the ODE stiffness and linearization error that this work targets. 
Third, we sweep the bearing noise standard deviation $\sigma_\theta$ over three levels (Table~\ref{tab:rmse}) with each level averaged over 100 Monte Carlo trials of 15 measurements over 15 seconds. The comparison baselines are EDH with $\Delta L$-adaptive Euler integration~\cite{mori2016adaptive}, the stochastic Gromov flow with $\Delta L$-adaptive Euler--Maruyama and optimized diffusion~\cite{dai2022design, daum2018new}, a bootstrap particle filter (PF), and the extended Kalman filter (EKF). The $\Delta L$ thresholds of the adaptive flows are calibrated so that, on average, they take the same number of substeps as the fixed-schedule NAEDH variants, enabling an equal-budget comparison. 

Fig.~\ref{fig:three-panel}b shows the convergence of the position error along the homotopy parameter $\lambda$ for the first update. At the fixed substep budget, NAEDH-ccr reaches a residual below $10^{-2}\,$m, an order of magnitude better than adaptive Euler at the same budget; even when adaptive Euler is granted a substantially larger budget, NAEDH-ccr still leads over most of the trajectory. Fig.~\ref{fig:three-panel}c traces the mean position error across the 15 filter updates: NAEDH-ccr and Gromov collapse the prior bias to $\sim 10^{-3}\,$m within the first measurement. NAEDH-lin (linear schedule) requires roughly seven updates to recover, while the EDH variants lag throughout. By $t = 15\,$s, all flow-based filters have converged to a process-noise-limited regime. Table~\ref{tab:rmse} summarizes RMSE and per-update runtime across the three measurement-noise levels. NAEDH-ccr is the most accurate filter at every noise level, with a per-update cost essentially equal to NAEDH-lin and EDH-adaptive and roughly half that of the stochastic Gromov flow. The bootstrap PF does not reach the flow-based accuracy even at the larger particle count, and the EKF diverges in the highly informative low-noise regime where the bearing measurement is most nonlinear.
  \begin{figure}[t]
    \centering
    \includegraphics[width=\columnwidth]{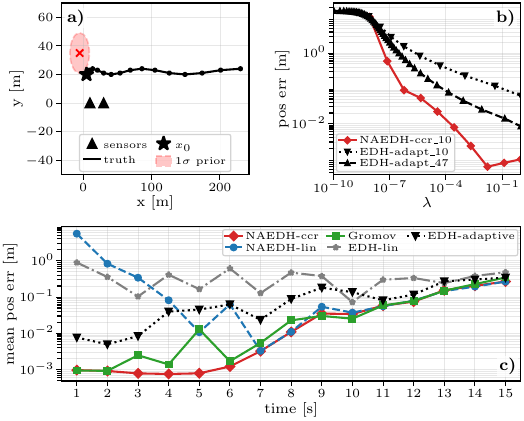}
    \caption{(a) Scenario: two static bearing sensors, biased prior mean with $1\sigma$ ellipse, and ground-truth target trajectory. (b) Position error [m] along the homotopy parameter $\lambda$ for the first update; subscripts denote the substep count $N$. (c) Mean position error [m] across the 15 filter updates for the compared methods averaged over 100 Monte Carlo runs for the $\sigma_\theta=0.001\deg$ setup.}
    \label{fig:three-panel}
  \end{figure}

\begin{table}[t]
\centering
\caption{RMSE (mean $\pm$ std) and per-update runtime at three measurement-noise levels $\sigma_\theta$.}
\label{tab:rmse}
\scriptsize
\setlength{\tabcolsep}{3pt}
\renewcommand{\arraystretch}{1.1}
\begin{tabular}{@{}l c ccc c@{}}
\toprule
 & & \multicolumn{3}{c}{RMSE [m]} & \\
\cmidrule(lr){3-5}
Method & $(N, N_p)$
       & $\sigma_\theta=0.001^\circ$ & $\sigma_\theta=0.05^\circ$ & $\sigma_\theta=1^\circ$
       & ms/update \\
\midrule
NAEDH-lin    & $(10,\,500)$          & $1.5 \pmstd{0.15}$ & $2.9 \pmstd{0.55}$ & $7.7 \pmstd{4.3}$ & $0.63$ \\
NAEDH-ccr    & $(10,\,500)$          & $\mathbf{0.11} \pmstd{0.055}$ & $\mathbf{1.8} \pmstd{0.40}$ & $\mathbf{3.6} \pmstd{1.9}$ & $0.62$ \\
\midrule[0.3pt]
EDH-adapt    & $(10,\,500)$          & $0.17 \pmstd{0.058}$ & $2.1 \pmstd{0.41}$ & $3.8 \pmstd{2.0}$ & $0.65$ \\
Gromov-adapt & $(10,\,500)$          & $0.13 \pmstd{0.065}$ & $2.3 \pmstd{0.98}$ & $3.7 \pmstd{2.0}$ & $1.1$ \\
Bootstrap PF  & $(-,\,10^4)$          & $47 \pmstd{23}$ & $48 \pmstd{23}$ & $22 \pmstd{15}$ & $0.97$ \\
Bootstrap PF  & $(-,\,10^5)$          & $46 \pmstd{21}$ & $35 \pmstd{18}$ & $9.2 \pmstd{6.6}$ & $10$ \\
EKF           & $(-,\,-)$             & diverged & $11 \pmstd{0.94}$ & $19 \pmstd{6.5}$ & $0.015$ \\
\bottomrule
\end{tabular}
\end{table}

\section{Conclusion}

This paper derives a closed-form, integration-free evaluation of the EDH particle flow update for vector linear Gaussian measurements and proves its equivalence to the Kalman posterior. Together with the proposed constant contraction rate schedule, the $N$-step piecewise variant gives a stiffness-robust, deterministic particle update for nonlinear measurement models. On a bearings-only tracking benchmark, NAEDH with the ccr schedule is the most accurate of the compared filters at a per-update cost comparable to existing flow-based methods.
\newpage
\bibliographystyle{IEEEtran}
\bibliography{ref}

\end{document}